\documentclass[10pt,twocolumn,twoside]{IEEEtran}
\usepackage{graphicx} 
\usepackage{subfigure}
\usepackage{amsmath}
\usepackage{amssymb}
\usepackage{bm}
\usepackage{color}
\usepackage{url}

\graphicspath{{./fig/}}

\newcommand{\qa}{{\bf a}}
\newcommand{\qb}{{\bf b}}
\newcommand{\qc}{{\bf c}}
\newcommand{\qd}{{\bf d}}

\newcommand{\qh}{{\bf h}}

\newcommand{\qn}{{\bf n}}

\newcommand{\qq}{{\bf q}}
\newcommand{\qr}{{\bf r}}

\newcommand{\qx}{{\bf x}}
\newcommand{\qy}{{\bf y}}
\newcommand{\qz}{{\bf z}}

\newcommand{\qA}{{\bf A}}

\newcommand{\qD}{{\bf D}}

\newcommand{\qH}{{\bf H}}
\newcommand{\qI}{{\bf I}}

\newcommand{\qQ}{{\bf Q}}
\newcommand{\qR}{{\bf R}}
\newcommand{\qS}{{\bf S}}

\newcommand{\qU}{{\bf U}}

\newcommand{\qW}{{\bf W}}
\newcommand{\qX}{{\bf X}}

\newcommand{\tS}{{\tt S}}

\newcommand{\tD}{{\tt D}}

\newcommand{\tE}{{\tt E}}

\newcommand{\qzero}{{\bf 0}}

\newcommand{\tr}{\mbox{trace}}

\newcommand{\be}{\begin{equation}} \newcommand{\ee}{\end{equation}}
\newcommand{\bea}{\begin{eqnarray}} \newcommand{\eea}{\end{eqnarray}}
\newtheorem{theorem}{Theorem}
\newtheorem{lemma}{Lemma}
\newtheorem{proposition}{Proposition}
\newtheorem{corollary}{Corollary}

\begin{document}
\title{Improving Physical Layer Secrecy Using Full-Duplex Jamming Receivers }
 \markboth{\textit{A Manuscript Accepted in The IEEE Transactions on Signal Processing } }{} 
\author{Gan Zheng, {\it Senior Member, IEEE}, Ioannis Krikidis, {\it Senior Member, IEEE}, Jiangyuan Li, {\it Member, IEEE}, Athina P. Petropulu, {\it Fellow, IEEE},
 and Bj$\ddot{\rm o}$rn Ottersten, {\it Fellow, IEEE}
\thanks{Copyright (c) 2012 IEEE. Personal use of this material is permitted. However, permission to use this
material for any other purposes must be obtained from the IEEE by sending a request to pubs-permissions@ieee.org.}
\thanks{Gan Zheng and Bj$\ddot{\rm o}$rn Ottersten are with the
Interdisciplinary Centre for Security, Reliability and Trust (SnT),
  University of Luxembourg, 4 rue Alphonse Weicker,  L-2721
Luxembourg, E-mail: {\sf\{gan.zheng,   bjorn.ottersten\}@uni.lu.}}
\thanks{Ioannis Krikidis is with the Department of Electrical and Computer Engineering, University of Cyprus, Cyprus, E-mail: {\sf krikidis@ucy.ac.cy}.}
\thanks{Jiangyuan Li and Athina P. Petropulu are with the Department of
Electrical and Computer Engineering, Rutgers-The State University of New Jersey, New Brunswick, NJ 08854 USA, E-mail:{\sf
\{jiangyli,athinap\}@rci.rutgers.edu.}} }

\maketitle

\begin{abstract}
This paper studies secrecy rate optimization in a wireless network with a single-antenna  source, a multi-antenna destination and a multi-antenna
eavesdropper. This is an unfavorable scenario for secrecy performance as the system is interference-limited. In the literature, assuming that the
receiver operates in half duplex (HD) mode,  the aforementioned problem has been addressed via use of cooperating nodes who act as jammers to
confound the eavesdropper. This paper investigates an alternative solution, which assumes the availability of a full duplex (FD) receiver. In
particular, while receiving data, the  receiver transmits jamming noise to degrade the eavesdropper channel. The proposed self-protection scheme
eliminates the need for external helpers and provides system robustness. For the case in which global channel state information is available, we aim
to design the optimal jamming covariance matrix  that maximizes the secrecy rate and mitigates loop interference associated with the FD operation. We
consider both fixed and optimal linear receiver design at the destination, and show that the optimal jamming covariance matrix is rank-1, and can be
found via an efficient 1-D search. For the case in which only statistical information on the eavesdropper channel is available,  the optimal power
allocation is studied in terms of ergodic and outage secrecy rates. Simulation results verify the analysis and demonstrate substantial performance
gain over conventional HD operation at the destination.
 \end{abstract}

\begin{center}
\end{center}

\section{Introduction}

Wireless communication technology has been integrated in almost all the aspects of social life. Cellular mobile networks, sensor/body networks, smart
grid, smart home and smart cities are just some examples of wireless systems that people are using or will use in the near future.  Under this
uncontrolled growth of wireless personal information transfer, confidentiality and secret transmission is introduced as an emergent research topic.
Traditionally, security is addressed at the upper layers of the protocol stack by using cryptographic tools,  which basically rely on  the
computational limitations of the eavesdroppers. Given that these approaches are sensitive to the wireless transmission and management of the secret
keys, recently there has been growing interest to ensure secrecy and confidentiality at the physical (PHY) layer. PHY layer security is an
information-theoretic approach and achieves secrecy by using channel codes and signal processing techniques. The seminal work of Wyner \cite{WYN}
introduced the degraded wiretap channel and  the fundamental notion of secrecy capacity. Since then, several studies have been proposed in the
literature from the viewpoints of both information theory \cite{Hero-03}--\cite{BEL} and signal processing \cite{artificial_noise}--\cite{DIN2}.

An efficient way to increase the secrecy rate in wireless systems is to degrade the decoding capability of the eavesdroppers by introducing controlled
interference, or artificial noise (AN). When the transmitter has multiple antennas, this can be achieved
by having the transmitter embed in its transmission
 artificial noise,  \cite{artificial_noise}\cite{Swindlehurst_09}\cite{KenMa-11},  which can be designed to avoid  the legitimate receiver and only  affect the
eavesdroppers  \cite{artificial_noise}\cite{Swindlehurst_09}. Under imperfect eavesdroppers channel state information (CSI), an AN-aided outage secrecy-rate
maximization problem was tackled in \cite{KenMa-11}.

When the transmitter is restricted to the use of one antenna,  a bank of external relays can be employed to collaboratively send jamming signals to
degrade the eavesdropper channel. This approach is referred to as cooperative jamming (CJ) \cite{Dong_09}--\cite{DIN2}.  The optimal CJ relay weights
design  for maximizing the secrecy rate
 is investigated in  \cite{Dong_09}\cite{Zheng_TSP_11}.  An opportunistic selection of two relays, where one relay re-forwards the
transmitted signal, while the other uses the CJ strategy is discussed in \cite{KRI} in the context of a  multi-relay network. In \cite{VIL}, the
authors   study the secrecy outage probability using CJ for different levels of CSI. The optimal transmit beamforming together with AN design  for
minimizing the secrecy outage probability is addressed in \cite{Jorswieck-12}\cite{Petropulu-Li-12}. The work in \cite{DIN2}  combines CJ with
interference alignment. The idea of using destination and source as jammers in the first phase of a two-phase relay network, is proposed in
\cite{Swindlehurst-11}. A destination-assisted jamming scheme is used  in \cite{Liu-TIFS13} to prevent the system becoming interference-limited.

Based on the existing literature, CJ approaches are mainly rely on external helpers, thus suffer from issues related to helper mobility,
synchronization and trustworthiness.  More recently, some approaches have been proposed that do not require external helpers for jamming, such as the
{\it iJam} scheme of
 \cite{KAT}, in which   the receiver  acts as a jammer. In the {\it iJam} scheme, the source repeats the transmission, while the receiver randomly jams one of the transmitted copies in each
sample time; since the eavesdropper does not know which sample is ``clean'', it cannot decode the transmitted signal. However, this self-protection
procedure requires a retransmission of the source signal which lowers throughput. In the majority of the literature, the terminals operate in
half-duplex (HD) mode, thus are not able to receive and transmit data simultaneously. However, recent advances on electronics, antenna technology and
signal processing allow the implementation of full duplex (FD) terminals that can receive and transmit data in the same time and on the same
frequency band. When it comes to relaying,  the FD operation can utilize the channel more effectively by achieving end-to-end transmission in one
channel use, as long as the  loop interference (LI) that leaks from the relay output to the relay input can be addressed \cite{RII1,KRI1}. Antenna
isolation, time cancelation and spatial precoding have been proposed in the literature for the mitigation of LI \cite{RII2,DAY,DUA,Viberg10}.
Comparison of the HD and the FD systems with transmit power adaptation is given in \cite{Taneli-Hybrid}.

In the context of PHY layer secrecy, the potential benefits of the FD relaying technology have not yet fully explored. In \cite{MUK} the authors
employ FD technology in a PHY layer secrecy context but from the adversary point of view; the paper investigates an FD eavesdropper with LI  that
optimizes its beamforming weights in order to minimize the secrecy rate of the system.
 An FD receiver generating AN is proposed in \cite{Xiong-12} to impair the eavesdropper's channel. The secrecy performance
 of that method is evaluated based on the outage secrecy region
  from a geometrical perspective.  However, in \cite{Xiong-12} it is assumed that the LI
 can be perfectly canceled at the receiver, which might be too optimistic.
 To the best of our knowledge, the use of FD with spatial LI mitigation in order to enlarge the secrecy rate of the system has not been reported in the literature.

Inspired by the works in \cite{Liu-TIFS13,KAT}, \cite{MUK} and \cite{Xiong-12}, in this paper we study the potential benefits of an FD destination
node simultaneously acting as a jammer and a receiver, with the goal of improving the secrecy rate. We consider  an unfavorable situation of a
single-antenna source, thus source generated    will not
  improve the secrecy rate a \cite{artificial_noise}\cite{Xiong-12} as there are not enough degrees of freedom to design the jamming signal.
 The proposed approach provides a self-protection mechanism at the receiver side without requiring external assistance, out-of-band
channel or data retransmission, and  is mainly of interest in applications in which  assisting nodes are not available and/or are not trusted.
Remarkably, for the case in which the destination has multiple transmit or receive antennas, we show that the system is no longer interference-limited and the secrecy
rate does not suffer from saturation at  high signal-to-noise-ratio (SNR) as in the HD case. In order to tackle the
problem of LI and maximize the secrecy rate, joint transmit and receive beamforming design is studied at the destination. Our contributions are summarized as follows:
 \begin{itemize}
    \item In contrast to previous works employing  FD nodes \cite{Xiong-12},  we do not assume complete self-interference cancelation, but we rather
      employ a LI model whose parameter describes the effect of the passive self-interference suppression.

    \item For the scenario of a  destination with one transmit and  one receive antenna,  and  a single antenna eavesdropper, we derive the closed form
    solution for the power allocation at the receiver. It is shown that due to the LI, the destination usually does not use all the available power.

    \item When the destination has multiple transmit antennas,  the system is no longer
    interference-limited. We show that in that case  the optimal jamming covariance matrix  is rank-1; and  propose efficient algorithms to find the optimal covariance matrix.
    Both fixed and optimal linear receivers are considered.
    \item When only eavesdropper CDI is available, we optimize the power allocation with respect to ergodic secrecy rate and outage secrecy rate.
 \end{itemize}

The organization of the paper is as follows. In Section II, we present the general system model and introduce the FD receiver and the self-jamming
operation. Section III narrows down to a single-antenna case and derives the optimal power allocation for the destination and the source.  Section IV
deals with optimal jamming covariance design. In Section V, we address the jamming design and power allocation to maximize the ergodic secrecy rate
and the outage secrecy rate. In Section VI, we study the optimal jamming covariance design when the eavesdropper knows the FD operation of the
destination and also adopts the optimal linear receiver. Simulation results are presented in Section VII and Section VIII concludes this paper.

\subsection{Notation}
Throughout this paper, the following notation will be adopted. Vectors and matrices are represented by boldface lowercase and uppercase letters,
respectively.  $\|\cdot\|$ denotes the Frobenius norm. $(\cdot)^\dag$ denotes the Hermitian operation of a vector or matrix. $\mathbb{E}[\cdot]$
denotes the expectation of a random variable. The notation $\qA\in \mathbb{C}^{M\times N}$ indicates that $\qA$ is complex matrix with dimensions
$M\times N$. $\qA\succeq \qzero$ means that $\qA$ is positive semi-definite. $\qI$ denotes an identity matrix of appropriate dimension. Finally,
${\bf x}\sim\mathcal{CN}({\bf m},{\bf\Theta})$ denotes a vector $\qx$ of complex Gaussian elements with a mean vector of ${\bf m}$ and a covariance
matrix of ${\bf\Theta}$.

\section{System Model}
 Consider a wireless communication system with one  source ${\tt S}$ with a single antenna, one destination ${\tt D}$ and
 one passive eavesdropper ${\tt E}$, with $M$ and $M_e$ antennas, respectively, as depicted in Fig. \ref{fig:sys}. $\tD$'s total  $M$ antennas are divided to $M_r$ receive
 antennas and $M_t$ transmit antennas with $M=M_t + M_r$.

Let $\qh_{sd}\in\mathbb{C}^{M_r\times1}$, $\qh_{se}\in\mathbb{C}^{M_e\times1}$ and $\qH_{ed} \in \mathbb{C}^{M_e\times M_t}$  denote the
 $\tS-\tD$, $\tS-\tE$ and between $\tD-\tE$ channels, respectively.
 $\qn_D\sim\mathcal{CN}(\qzero,\qI)$ and $\qn_E\sim\mathcal{CN}(\qzero,\qI)$ represent noise  at $\tD$ and $\tE$, respectively.
The transmit signal $s$ is assumed to be  a  zero-mean complex Gaussian random variable with power constraint ${\tt E}[|s|^2]\le P_s$.

In the considered network configuration, there is no external relay or helper to assist $\tD$ against $\tE$. Instead, as shown in Fig. \ref{fig:sys},
$\tD$ helps itself by operating in  FD mode and transmitting jamming signals to degrade the  quality of the eavesdropper's link to $\tS$.
 The receiver transmits a jamming signal while it simultaneously receives
the source signal.  This creates a feedback loop channel  from the relay output to the relay input through the effective channel $\qH_{si} =
\sqrt{\rho} \qH\in \mathbb{C}^{M_r
\times M_t}$, where $\qH$ is 
a fading loop channel   \cite{RII1,KRI1,RII2,DAY}. In order to make our study more general, we parameterize the LI channel by introducing the
variable $\rho$ with $0 \leq \rho \leq 1$ \cite{MUK}; this parameter models the effect of  passive LI suppression such as antenna isolation
\cite{RII2}. Therefore, $\rho=0$ refers to the ideal case of no LI (perfect antenna isolation) while $0<\rho \leq 1$ corresponds to different LI
levels.

We assume that $\tD$ transmits the jamming signal $\qn\sim\mathcal{CN}(\qzero,\qQ)$. The covariance matrix $\qQ$ will be designed to a maximize the
secrecy rate  with power constraint $p_d=\tr(\qQ)\le P_d$. We assume that $\tE$ is not aware of the FD operation of $\tD$, and simply uses a maximum
ratio combining  (MRC)  receiver, $\qh_{se}^\dag$\footnote{ Once $\tE$ knows the FD mode of $\tD$, it may adapt its strategy as well, e.g., it can
use the linear minimum mean square error (MMSE) receiver to mitigate the jamming from $\tD$. This will be analyzed in Section VI.}. The received
signal at $\tE$ is $\qy_E=\qh_{se} s +  \qH_{ed} \qn + n_e$. After applying MRC receiver, the data estimate at $\tE$ is: \be
    \hat s_e = \frac{\qh_{se}^\dag}{\|\qh_{se}\|}(\qh_{se} s +  \qH_{ed} \qn + n_e ).
\ee
 $\tD$   employs a linear receiver, $\qr$, on its received signal,  $\qy_D=\qh_{sd} s +  \sqrt{\rho}\qH \qn + n_d$, to obtain the data estimate
\be
    \hat s_d = \qr^\dag (\qh_{sd} s +  \sqrt{\rho}\qH \qn + n_d ), \mbox{~with~} \|\qr\|=1.
 \ee

 The  achievable secrecy rate is expressed as \cite{VIL}
 \bea\label{eqn:Rs}
    R_{S} &=&\max\left\{0, \log_2\left(1+ \frac{P_s |\qr^\dag\qh_{sd}|^2}{1+    \rho\qr^\dag\qH\qQ\qH^\dag\qr}\right) \right.\notag\\
  &&  \left.-\log_2\left(1+\frac{P_s\|\qh_{se}\|^2}{1+ \frac{\qh_{se}^\dag\qH_{ed}\qQ\qH_{ed}^\dag\qh_{se}}{\|\qh_{se}\|^2}}\right)\right\}.
\eea

We should note that the above secrecy rate also corresponds to the case of an external jammer helper with $M_t$ antennas,  whose channels to $\tD$
and $\tE$ are $\rho\qH$ and $\qH_{ed}$, respectively.
 The difference in {  the considered} case is that $\tD$   needs to perform LI cancelation to achieve that rate.
  In this paper, we consider LI mitigation in the spatial domain, in order to keep the complexity low.  For FD systems with multiple antennas,  the suppression of  LI in the spatial domain  has been addressed in the literature i.e.,
  \cite{RII2,Viberg10},
 where the low rank of the spatial LI channel is exploited to avoid the transmit signal noise,   which is the main source of the  residual LI.
 Alternatively, the LI can be mitigated in the time domain using analogue or/and digital LI cancelation. However, this approach requires expensive cancelation circuits and  is sensitive to  transmit noise due to  non-idealities at the relay
 node \cite{DAY}.
  It is worth noting that the time cancelation requires two rounds of pilots for the case of successive analogue/digital cancelation, which further increases the complexity of the suppression process \cite{DUA}.

Throughout this paper, we assume that positive secrecy rate is achievable. In the following, we will study the transmit beamforming optimization at  $\tD$ to
 maximize the secrecy rate $R_S$, with the power constraint $P_d$ at $\tD$,  or the joint $\tS-\tD$ power constraint $P_T$.  We will begin with perfect CSI, and later consider  the case
in which  the CSI of $\tE$ is not known.

\section{Perfect CSI with single-antenna terminals}
 In this section, we assume that global CSI is available, $\tE$ has a single receive antenna and $\tD$ has one transmit and one receive antenna.
In this case, the LI cannot be mitigated in the  spatial domain but is   controlled by the power control process. Although substantial isolation
 between the transmitter and receiver channels has been recently reported \cite{Isolation}, the single-antenna case is mainly used as a baseline scheme for comparison reasons and as a guideline  for the multiple-antenna case. The main focus
of the work is the multiple-antenna case which allows LI mitigation in the spatial domain.

 Following the conventional notation, we use lower-case letters to denote scalar channels. Because the jamming signal from $\tD$   degrades the performance
 of both $\tE$ and $\tD$ via the LI channel, the destination needs to carefully choose the transmit power $p_d$ to achieve a good balance. Using (\ref{eqn:Rs}), the secrecy rate maximization via power control at
 $\tD$ can be formulated as
 \be \label{eqn:pd}
   \max_{ 0\le p_d\le P_d} ~ f_\rho(p_d)\triangleq\frac{1+ \frac{P_s |h_{sd}|^2}{1+  \rho p_d |h|^2 }}{1+\frac{P_s|h_{se}|^2}{1+  p_d
   |h_{ed}|^2}}.
 \ee
One may observe that (4) looks as if there was an external helper who transmits jamming signals to improve the secrecy rate.

 Before we derive its solution, we first study the conditions under which when positive secrecy rate is possible.
 \begin{lemma}
 The conditions under which the positive secrecy is achieved are
 \begin{itemize}
    \item $\rho< \min(\frac{|h_{ed}|^2}{|h|^2},\frac{ |h_{sd}|^2(1+  P_d|h_{ed}|^2)}{P_d |h|^2 |h_{se}|^2} -\frac{1}{P_d |h|^2 })$; or
    \item $\rho\ge   \frac{|h_{ed}|^2}{|h|^2}$ and $|h_{sd}|^2> |h_{se}|^2$.
 \end{itemize}
 \end{lemma}
 \begin{proof}
  The positive secrecy is achieved if and only if the optimal objective value of the following problem is strictly greater than 1:
      \be
   \max_{ 0\le p_d\le P_d}  \frac{ |h_{sd}|^2}{1+  \rho p_d |h|^2 } \frac{1+  p_d|h_{ed}|^2}{ |h_{se}|^2}.
 \ee
  It is not difficult to see that the optimal value is achieved at either $p_d^*=0$ or $p_d^*=P_d$.
  \begin{itemize}
    \item When $\rho<\frac{|h_{ed}|^2}{|h|^2}$, $p_d^*=P_d$. Positive secrecy rate is achieved if and only if
     \bea
    \frac{ |h_{sd}|^2}{1+  \rho P_d |h|^2 } > \frac{ |h_{se}|^2}{1+  P_d
   |h_{ed}|^2},  ~~\mbox{or}\\
   \rho< \frac{ |h_{sd}|^2(1+  P_d|h_{ed}|^2)}{P_d |h|^2 |h_{se}|^2} -\frac{1}{P_d |h|^2 }.
     \eea
    \item When $\rho\ge \frac{|h_{ed}|^2}{|h|^2}$,  $p_d^*=0$. Positive secrecy rate is achieved if and only if
        \be
        |h_{sd}|^2> |h_{se}|^2.
        \ee
  \end{itemize}
    This completes the proof.
 \end{proof}
  It is clearly seen that when the LI can be sufficiently suppressed, the positive secrecy region can be extended by the FD destination; when the LI
  is above a certain threshold, FD does not provide any performance gain therefore is not necessary.

 Given that a positive secrecy rate is achievable, the optimal solution to (\ref{eqn:pd}) is given in the Proposition \ref{prop:pd} and the proof is provided in Appendix A.
\begin{proposition}\label{prop:pd}
 Suppose  the roots of $f_\rho'(p_d)=0$ are $x_1(\rho), x_2(\rho)$. If they are both real, we assume $x_2(\rho)\ge x_1(\rho)$. If both are complex or non-positive, we define $x_2(\rho)=0$.
 Let  $\delta=\frac{|h_{sd}|^2|h_{ed}|^2}{|h_{se}|^2|h|^2}$. Then the optimal $p_d^*(\rho)$ to maximize $f_\rho(p_d)$ is given below:
         \begin{itemize}
           \item[i)]If $\rho< \min(\delta,1)$, $p_d^*(\rho) =\min(P_d, x_2(\rho))$;
           \item[ii)]If $1\ge\rho\ge\min(\delta,1)$, $p_d^*(\rho) =0$  or $p_d^*(\rho) =P_d$, whichever gives a higher objective value.
 \end{itemize}
\end{proposition}

As a special case,   when the $\tS-\tD$ and $\tS-\tE$ channels have the same strength or $\tS-\tE$ and $\tS-\tD$ have the same distances, we have the
following result.
\begin{corollary}
$ p_d(\rho)$ is   a monotonically non-increasing function when $|h_{sd}|^2=|h_{se}|^2$.
\end{corollary}
\begin{proof}
 When $|h_{sd}|^2=|h_{se}|^2$,  the equation (\ref{eqn:fxzero}) reduces to $(b-d)x^2 - (1/d-1/b)(1+a)=0$ and we have the  root
 $x_2(\rho)= \sqrt{\frac{1+a}{bd}}=  \sqrt{\frac{1+P_s}{ \rho|h|^2 |h_{ed}|^2}}$, where $a,b,c,d$ are defined in
 Appendix A. 
 \begin{itemize}
            \item[i)] When $\rho< \frac{|h_{ed}|^2}{|h|^2}$, $f_\rho^{'}(p_d)>0$ for $0\le p_d\le x_2(\rho)$, so $p_d^*(x) =\min(P_d, x_2(\rho))$;
            \item[ii)] When $\rho\ge \frac{|h_{ed}|^2}{|h|^2}$, $f_\rho^{'}(p_d)\le0$ for $0\le p_d\le x_2(\rho)$, so  $p_d^*(x) =0$. In this case, positive secrecy rate is not achievable.
 \end{itemize}
 Clearly $x_2(\rho)$ is a monotonically decreasing function and this completes the proof.
\end{proof}

Typical curves of $p_d$ are shown in Fig. \ref{fig:pd:nonmono}, where we plot the normalized $p_d(\rho)$ for two sets of  randomly generated channel
parameters. For each subfigure, we also plot the results by setting $|h_{sd}|^2=|h_{se}|^2=1$, while keeping the other parameters fixed. We can see
that in general the receiver may not always use full power and $p_d(\rho)$ is not necessarily monotonically changing with $\rho$.

This can be explained by the fact that when $\rho$ is small, the self-interference is also small, so the receiver can use full or high power to
confuse the eavesdropper; as $\rho$ increases, the receiver needs to reduce its transmit power in order not to generate too much self-interference;
when $\rho$ is very close to 1, the receiver causes high interference to both itself and the eavesdropper, but if the eavesdropper suffers more,
the receiver can still increase its transmit power, otherwise, it should decrease its power. The results in Fig. \ref{fig:pd:nonmono} verify our
analysis given in Corollary 1.

Next we consider the case where $\tS$ and $\tD$ have a total power constraint $P_T$ and we denote their power as $p_s$ and $p_d$, respectively. We
aim to maximize the secrecy rate below $$ f_\rho(p_s, p_d)\triangleq\frac{1+ \frac{p_s |h_{sd}|^2}{1+  \rho p_d |h|^2 }}{1+\frac{p_s|h_{se}|^2}{1+
p_d
   |h_{ed}|^2}}$$ by using optimal power allocation between $\tS$ and $\tD$.

\begin{proposition}
 If the source and the receiver has a total power constraint $P_T$ and strictly positive secrecy rate is achievable, then full power should be used, i.e.,
 $p_s + p_d=P_T$.
\end{proposition}
\begin{proof}
    Strictly positive secrecy rate implies that there exists a solution $(p_s, p_d)$ such that
    $$
 \frac{p_s |h_{sd}|^2}{1+\rho p_d |h|^2 }> \frac{p_s|h_{se}|^2}{1+p_d
   |h_{ed}|^2},
    $$
and this means $f(p_s, p_d)$ is increasing in $p_s$. So if there is any unused power $\Delta p$, we can always add it to $p_s$ to obtain
$f(p_s+\Delta p, p_d)>f(p_s, p_d)$.
\end{proof}
Consequently, the power allocation  problem   can be formulated as
\begin{align}
&\max_{0\le\alpha\le 1} \frac{1+\frac{\alpha P_T |h_{sd}|^2}{1+\rho (1-\alpha)P_T|h|^2 }}{1+\frac{\alpha P_T|h_{se}|^2}{1+(1-\alpha)P_T
   |h_{ed}|^2}}, \notag \label{opt2}
\end{align}
where we have assumed that $p_s=\alpha P_T$ and $P_d=(1-\alpha)P_T$ with $0 \leq \alpha \leq 1$.

For simplicity, we assume that $|h_{sd}|^2=|h_{se}|^2=1$, i.e., the channels to $\tD$ and $\tE$ have the same normalized strength, then we have the
following solution:
\begin{align}
\alpha^*= \left\{ \begin{array}{l} 0,\;\;\;\;\;\;\;\;\;\;\;\;\;\;\;\;\;\;\;\;\;\;\;\;\;\;\;\;\;\;\;\;\;\;\;\;\;\;\text{if}\;\rho\ge \frac{|h_{ed}|^2}{|h|^2};\\
\frac{1}{1+\sqrt{ \frac{P_T+1}{(P_T|h_{ed}|^2+1)(P_T |h|^2\rho +1)}}},\; \text{if} \;{0\le \rho< \frac{|h_{ed}|^2}{|h|^2}}. \end{array}  \right.
\end{align}

It can be seen that when $0\le \rho< \frac{|h_{ed}|^2}{|h|^2}$, i.e., positive secrecy rate is possible,   the source power $p_s$ is monotonically
increasing in $\rho.$ This is because the larger $\rho$ is, the less the receiver power $p_d$ is, and the conclusion follows due to constant sum
power of $\tS$ and $\tD$.

We have shown that for the single-antenna case, $\tD$ may not use full transmit power $P_d$. To conclude this section, we study the performance at high signal-to-noise (SNR). Suppose $\frac{p_s}{p_d}=\beta$ remains constant, then
 \be
\lim_{p_d\rightarrow\infty} R_{S}  = \max\left(0, \log\left(1+ \frac{\beta|h_{sd}|^2}{\rho |h|^2 }\right) -
\log\left(1+\frac{\beta|h_{se}|^2}{|h_{ed}|^2}\right)\right),
 \ee
which indicates the secrecy rate will saturate when $p_d\rightarrow\infty$. This is because the self-interference cannot be fully canceled \cite{DAY}
and becomes a limiting factor at high SNR.

\section{Multiple-antenna Receiver}
 In this section, we study the transmitter design when both the destination and the eavesdropper has multiple antennas.
 Based on (\ref{eqn:Rs}), we first formulate the secrecy rate maximization problem with power constraint at $\tD$, i.w.,
  \bea\label{eqn:prob:0}
    \max_{\qQ,\|\qr\|=1} && \frac{ 1+ \frac{P_s |\qr^\dag\qh_{sd}|^2}{1+
    \rho\qr^\dag\qH\qQ\qH^\dag\qr}} {1+\frac{P_s\|\qh_{se}\|^2}{1+\frac{\qh_{se}^\dag\qH_{ed}\qQ\qH_{ed}^\dag\qh_{se}}{\|\qh_{se}\|^2}}} \\
    \mbox{s.t.} && \qQ\succeq \qzero,  ~\tr(\qQ)\le P_d.\notag
 \eea

 Different from the single-antenna case, the secrecy rate can
 keep increasing with transmit power, as stated in Lemma \ref{lem:increase:rate} below.
 \begin{lemma}\label{lem:increase:rate}
   Given  $M_t>1$ or $M_r>1$,  the system is not interference-limited.
 \end{lemma}

 \begin{proof}
 It suffices to show that there exists a scheme whose achievable secrecy rate  does not saturate. It is easily seen that as long as $\tD$ has either multiple transmit or receive antennas,
  $\qr$ and $\qQ$ can be chosen according to the zero-forcing (ZF) criterion such that $\qr^\dag\qH\qQ\qH^\dag\qr=0$ but
 $\qh_{se}^\dag\qH_{ed}\qQ\qH_{ed}^\dag\qh_{se}>0$. As a result, the secrecy rate becomes
 \bea
    R_{S,ZF}    &=&
    \max\left(0, \log\left(1+  P_s \|\qh_{sd}\|^2 \right) \right.\notag\\
   && \left. -
    \log\left(1+\frac{P_s\|\qh_{se}\|^2}{1+ \frac{ \qh_{se}^\dag\qH_{ed}\qQ\qH_{ed}^\dag\qh_{se}}{\|\qh_{se}\|^2}}\right)\right).
\eea
 It is noted that the jamming signal sent by $\tD$ does not affect itself but degrades the received signal-to-noise plus interference (SINR) at
 $\tE$. Therefore, $R_{S,ZF}$ is a strictly monotonically increasing function of $\tr(\qQ)$. The system is no longer interference-limited,
  and  the secrecy rate $R_{S,ZF}$ can increase with $P_d$ without saturation. This completes the proof.
 \end{proof}

In the following, we assume that $M_t>1$ and derive the optimal solutions to problem (\ref{eqn:prob:0}). When  $M_t=1$, the jamming covariance matrix
design reduces to jamming power optimization, and has been addressed in the previous section. Before we address the optimal jamming covariance
design, we will present two useful lemmas about the properties of the optimal $\qQ^*$.
\begin{lemma}\label{lemma:rank}
   The optimal $\qQ^*$ to solve (\ref{eqn:prob:0}) should be rank-1.
 \end{lemma}
 \begin{proof}
    Let $\qQ^*, \qr^*$ be any solution to the problem (\ref{eqn:prob:0}).
If $\qQ^*$  is rank one, the desired result is obtained. If $\qQ^*$ is not rank-1, let $\qr^\dag\qH\qQ^*\qH^\dag\qr = x$,
$\qh_{se}^\dag\qH_{ed}\qQ^*\qH_{ed}^\dag\qh_{se}=y$ and consider the following problem \be\label{eqn:Q00} \min_{\qQ}\ \qr^\dag\qH\qQ\qH^\dag\qr, \quad
\mathrm{s.t.}\quad \mathrm{Tr}(\qQ\qH_{ed}^\dag\qh_{se}\qh_{se}^\dag\qH_{ed}) = y, \ \mathrm{Tr}(\qQ) \le P_d. \ee Obviously, $\qQ^*$ is feasible for
the above problem. We want to prove that  $\qQ^*$ is also the optimal solution to (\ref{eqn:Q00}). Let $\qQ'$ be any solution to the above problem
(\ref{eqn:Q00}). Then it holds that $\qr^\dag\qH\qQ'\qH^\dag\qr = x$. This is because if $\qr^\dag\qH\qQ'\qH^\dag\qr < x$, then $\qQ'$ is feasible
for the problem of (\ref{eqn:prob:0}) but achieves a strictly larger objective value than $\qQ^*$ in (\ref{eqn:prob:0}). But this contradicts the
fact that  $\qQ^*$ is the optimal solution to (\ref{eqn:prob:0}), so $\qQ^*$ is also the optimal solution to (\ref{eqn:Q00}).

Problem (\ref{eqn:Q00}) is a homogeneous quadratically constrained quadratic program (QCQP) with two constraints. According to the results in
\cite{Huang-07}, it  has a rank-1 solution. Thus it follows that the problem of (\ref{eqn:prob:0}) has a rank-1 solution $\qQ^*$.
 \end{proof}
 The next lemma is about the power consumption of jamming signals.
\begin{lemma}\label{lemma:power}
   Given that $\qH^\dag \qr$ does not align with $\qH_{ed}^\dag\qh_{se}$, i.e., there exists no scalar $g$ such that
   $\qH^\dag \qr = g \qH_{ed}^\dag\qh_{se}$, the optimal $\qQ^*$  satisfies  $\tr(\qQ^*)=P_d$.
 \end{lemma}
 \begin{proof}
 We prove it  by contradiction. Suppose the optimal solution   is $\qQ_1$ and $\tr(\qQ_1)<P_d$.
    We then choose $\qx$ such that $\qh_{se}^\dag\qH_{ed}\qx=0$,  $\qr^\dag\qH\qx\ne 0$ and $\tr(\qQ_1)+\|\qx\|^2=P_d$.
    This is possible   due to the assumption that $\qH^\dag \qr$ does not align with $\qH_{ed}^\dag\qh_{se}$.
    Construct a new solution $\qQ_2 = \qQ_1 +  \qx\qx^\dag$. It is easy to verify that
    $\qQ_2$ is a strictly better solution than $\qQ_1$ with $\tr(\qQ_2)=P_d$, which contradicts the optimality of $\qQ_1$. This completes the proof.
\end{proof}

 In the special case in which the condition   $\qH^\dag \qr = g \qH_{ed}^\dag\qh_{se}$ holds,  the equivalent channels for the two links $\tD-\tD$ and $\tD-\tE$ align with each other. This is not desired and can normally be avoided by $\tD$ via proper
  design of $\qr$.
 Because this condition is easy to detect, in the sequel we assume it does not
 hold.

 By introducing $\qQ=P_d\qq\qq^\dag$ and $\|\qq\|=1$, problem (\ref{eqn:prob:0}) becomes
  \bea\label{eqn:prob:3}
    \max_{\|\qq\|=1,\|\qr\|=1} && \frac{ 1+ \frac{P_s |\qr^\dag\qh_{sd}|^2}{1+
    \rho P_d\qr^\dag\qH\qq\qq^\dag\qH^\dag\qr}}
    {1+\frac{P_s\|\qh_{se}\|^2}{1+\frac{P_d\qh_{se}^\dag\qH_{ed}\qq\qq^\dag\qH_{ed}^\dag\qh_{se}}{\|\qh_{se}\|^2}}}.
 \eea

 Next we will study problem   (\ref{eqn:prob:3}).  We begin with a fixed receiver and then move to
the optimal linear receiver; the analysis is also extended for the case of joint power allocation at $\qS$ and $\qD$.

\subsection{Optimal Solution with a Fixed Receiver}
 Let us assume that the receiver $\qr$ is fixed and independent of $\qQ$. Possible choices include the MRC and the minimum mean square error (MMSE) receivers.

 The problem of (\ref{eqn:prob:3}) is complicated and difficult to solve it directly. Instead, we solve the following problem by introducing an auxiliary
 variable $t$:
 \bea\label{eqn:gt}
  \max_{\qq} && \qh_{se}^\dag\qH_{ed} \qq\qq^\dag\qH_{ed}^\dag\qh_{se}\\
    \mbox{s.t.} && \qr^\dag\qH \qq\qq^\dag\qH^\dag\qr =   t, \|\qq\|=1.\notag
 \eea
By denoting its optimal objective value as $g(t)$, the original problem of (\ref{eqn:prob:3}) becomes
 \be\label{eqn:opt:t}
    \max_{t\ge0} f(t) \triangleq    \frac{ 1+ \frac{P_s |\qr^\dag\qh_{sd}|^2}{1+    \rho t }}
    {1+\frac{P_s\|\qh_{se}\|^2}{1+\frac{P_d g(t)}{\|\qh_{se}\|^2}}}.
 \ee
 We develop the following intermediate results to efficiently solve the problem of (\ref{eqn:gt}).
 \begin{lemma}\label{lem:gt}
   Let $\qc_2 = \frac{\qH_{ed}^\dag\qh_{se}}{\|\qH_{ed}^\dag\qh_{se}\|}$, $\qc_1 = \frac{\qH^\dag\qr}{\|\qH^\dag\qr\|}$ and $r = |\qc_1^\dagger\qc_2|$. Then $g(t)  =  1 - (r\sqrt{1-t} - \sqrt{(1-r^2)t})^2$  and  is a concave function in $t$.
 \end{lemma}
 \begin{proof}
 With the defined notation, (\ref{eqn:gt}) becomes
 \bea\label{eqn:lem5}
  \max_{\qq}&&  \qq^\dag \qc_2\qc_2^\dag\qq, \notag\\
  \mathrm{s.t.}&& \qq^\dag \qc_1\qc_1^\dag\qq = t, \quad \|\qq\|=1.
 \eea
The closed-form solution $g(t)$ then follows from Lemma 2 in \cite{Li-11}. The concavity can be proven by confirming that the second order derivative
of $g(t)$ is negative.
 \end{proof}
 Note that a similar  problem as (\ref{eqn:lem5}) has been studied in \cite[(14)]{Zheng_TSP_11} where   individual constraints on each element of
 $\qq$ are assumed, hence, a closed-form solution  is not possible. The closed-form solution in Lemma   \ref{lem:gt} speeds up the algorithm to solve (\ref{eqn:opt:t}).

 Given Lemma \ref{lem:gt}, we have the following theorem that can be used to develop efficient algorithms to solve (\ref{eqn:opt:t}).
\begin{theorem}
$f(t)$ is quasi-concave in $t$ and its maximum can be found via bisection search.
\end{theorem}
\begin{proof}
    Given that $g(t)$  is a concave function, the result follows from Theorem 3 in \cite{Zheng_TSP_11}.
\end{proof}

\subsection{Optimal Solution with The Optimal Linear Receiver}
 In this section, we aim to jointly optimize $\qq$ as well as   the receiver design at $\tD$.
From (\ref{eqn:Rs}),  the optimal linear    receiver to maximize the received SINR at $\tD$ is given by
 \be\label{eqn:mmse:r}
    \qr = \frac{\left(\rho\qH\qQ\qH^\dag+ \qI\right)^{-1}\qh_{sd}  }{\|\left(\rho\qH\qQ\qH^\dag+ \qI\right)^{-1}\qh_{sd}\|}.
 \ee
 Then the achievable secrecy rate is expressed as
 \bea\label{eqn:RS:q}\small
    R_{S} &=&\max\left(0, \log_2\left(1+ P_s \qh_{sd}^\dag\left(\rho\qH\qQ\qH^\dag+ \qI\right)^{-1}\qh_{sd}  \right)\right.\notag\\&&\left. -
    \log_2\left(1+\frac{P_s\|\qh_{se}\|^2}{1+ \frac{\qh_{se}^\dag\qH_{ed}\qQ\qH_{ed}^\dag\qh_{se}}{\|\qh_{se}\|^2}}\right)\right)\notag\\
    &=&\max\left(0, \log_2\left(1+ P_s \|\qh_{sd}\|^2  -  \frac{\rho P_s P_d|\qh_{sd}^\dag\qH\qq|^2}{1+ \rho P_d\qq^\dag\qH^\dag\qH\qq} \right)
   \right.\notag\\ &&  \left. -\log_2\left(1+\frac{P_s\|\qh_{se}\|^2}{1+ \frac{P_d|\qh_{se}^\dag\qH_{ed}\qq|^2}{\|\qh_{se}\|^2}}\right)\right),
 \eea
  where we have used the rank-1 property of $\qQ$ and the notation $\qQ=P_d\qq\qq^\dag$ and $\|\qq\|=1$, together with matrix inversion lemma.

 The secrecy rate expression $R_S$  in (\ref{eqn:RS:q}) is still complicated. To tackle it, we study the following problem with  the parameter $t$:
 \bea\label{eqn:hq}
    \max_{\qq} && |\qh_{se}^\dag\qH_{ed} \qq|^2\\
    \mbox{s.t.} &&\frac{|\qh_{sd}^\dag\qH\qq|^2}{1+ \rho P_d\qq^\dag\qH^\dag\qH\qq}= t,   ~\|\qq\|^2= 1,\notag
 \eea
 which is a nonconvex quadratic optimization problem and difficult to solve, instead, we first study a modified problem below
 by introducing $\tilde\qQ=\qq\qq^\dag$:
  \bea\label{eqn:ht}
    \max_{\tilde\qQ } && \tr(\qh_{se}^\dag\qH_{ed}\tilde \qQ \qH_{ed}^\dag \qh_{se})\\
    \mbox{s.t.} && \tr(\tilde\qQ (\qH^\dag\qh_{sd}\qh_{sd}^\dag\qH - t\rho P_d \qH^\dag\qH )   )= t,\notag\\
    &&\tilde\qQ\succeq \qzero,   ~\tr(\tilde\qQ)= 1.\notag
 \eea
 (\ref{eqn:ht}) is a semidefinite programming problem and  the method to solve (\ref{eqn:ht}) is provided  in Appendix B. Note that in Appendix B,
 we have shown that the optimal $\tilde \qQ$ should be rank-1, so (\ref{eqn:hq}) and (\ref{eqn:ht}) are equivalent in the sense that given the optimal solution
 $\tilde \qQ^*$ to  (\ref{eqn:ht}), the optimal $\qq^*$ to solve (\ref{eqn:hq}) can be extracted via $\tilde\qQ^*=\qq^*{\qq^*}^\dag$ .
 Denote its optimal objective value as $h(t)$,   the secrecy rate maximization problem can be formulated as
  \bea\label{eqn:Rt}
   \max_{t\ge 0} && R(t) \triangleq\log \left( \frac{1+ P_s \|\qh_{sd}\|^2  -  \rho P_s P_d t }{1+\frac{P_s\|\qh_{se}\|^2}{1+ \frac{P_d h(t) }{\|\qh_{se}\|^2}}}
   \right).
 \eea
 Thus the maximum of  $R(t)$  can be found via a one-dimensional search.

\subsection{Joint $\tS$-$\tD$ Power Allocation}
\subsubsection{Optimal Solution}
 If $\tD$ and $\tS$ can share a total power $P_T$, then the problem is revised to
  \be\label{eqn:total:pow}
    \max_{\qQ\succeq \qzero, p_s \ge 0,  ~ p_s + \tr(\qQ)\le P_T} ~~ R_S,
 \ee
 where we have used $p_s$ to denote the source power as a variable.
  Since we have derived the solution when there is only power constraint $P_d$ on $\tD$ with a parameter $t$, we can perform a 2-D search over $p_s$ and $t$ to find the optimal $\qQ$ and power allocation $p_s$.
   This is because all power should be used up at  the optimum, i.e., $p_s + \tr(\qQ)=P_T$.

 For the fixed linear receiver, we can reduce 2-D search to a  1-D search. To achieve that, we first define a new objective function based on (\ref{eqn:opt:t}):
 \be\label{eqn:fpt}
  f(p_s, t) \triangleq    \frac{ 1+ \frac{p_s |\qr^\dag\qh_{sd}|^2}{1+    \rho t }}
    {1+\frac{p_s\|\qh_{se}\|^2}{1+\frac{(P_T-p_s) g(t)}{\|\qh_{se}\|^2}}}.
 \ee
 For fixed $t$, setting its derivative   regarding   $p_s$ to be zero leads to the following quadratic equation
 \be\label{eqn:eqn}
    A p_s^2 + B p_s + C=0,
 \ee
 where 
$A\triangleq - \frac{ g(t)}{\|\qh_{se}\|^2}(\|\qh_{se}\|^2-\frac{ g(t)}{\|\qh_{se}\|^2}),
 B\triangleq-\frac{ g(t)}{\|\qh_{se}\|^2}(1+P_T \frac{ g(t)}{\|\qh_{se}\|^2}), C\triangleq (1+P_T \frac{ g(t)}{\|\qh_{se}\|^2})
 ((1+P_T \frac{ g(t)}{\|\qh_{se}\|^2}) - \frac{(1+    \rho t)\|\qh_{se}\|^2 }{ |\qr^\dag\qh_{sd}|^2} )$.
 The optimal $p_s^*$ should either be $P_T$ or one root of the equation (\ref{eqn:eqn}) and denote it as $P_S(t)$.
 Then the total power constrained secrecy optimization problem  (\ref{eqn:total:pow}) for a fixed receiver becomes
 \be
    \max_{t\ge 0} ~~   \frac{ 1+ \frac{P_S(t) |\qr^\dag\qh_{sd}|^2}{1+    \rho t }}
    {1+\frac{P_S(t)\|\qh_{se}\|^2}{1+\frac{(P_T-P_S(t)) g(t)}{\|\qh_{se}\|^2}}}.
 \ee
 Its optimal solution can be obtained via 1-D optimization over $t$ only.

 For the optimal linear receiver, it is not possible to apply this procedure because $h(t)$, obtained from (\ref{eqn:ht}), is also a complex function of $p_d$.
 Due to this non-separability,  we have to use 2-D search to find the optimal power allocation.
\subsubsection{ZF Solution} In this section, we study a simple closed-form solution  based on  the  ZF criterion. For conciseness,
we only consider the case of the optimal receiver. This corresponds to $t=0$ in the problem  of (\ref{eqn:Rt}).

 In (\ref{eqn:RS:q}), we can see that $\qh_{sd}^\dag\qH\qq$ is a self-interference term that may limit the system performance, so we
  impose an additional constraint that self-interference is zero, i.e.,  $\qh_{sd}^\dag\qH\qq=\qzero$. Using this condition together with the
  matrix inversion  lemma, we can derive that the optimal linear receiver in (\ref{eqn:mmse:r}) reduces to $\qr = \frac{\qh_{sd}}{\|\qh_{sd}\|}$, which is
  essentially the  MRC
  receiver. Therefore (\ref{eqn:RS:q}) is simplified to
 \bea\label{eqn:RS:q:zf}
    R_{S}     &=&
    \max\left(0, \log\left(1+  p_s \|\qh_{sd}\|^2 \right)-  \right.\notag\\
    &&\left.\log\left(1+\frac{p_s\|\qh_{se}\|^2}{1+ \frac{p_d |\qh_{se}^\dag\qH_{ed}\qq|^2}{\|\qh_{se}\|^2}}\right)\right), \|\qq\|=1.
\eea

 To maximize (\ref{eqn:RS:q:zf}), we first study a simple problem below:
 \bea\label{eqn:qq}
    \max_{\qq} && \|\qh_{se}^\dag\qH_{ed}\qq\|^2\\
    \mbox{s.t.} && \qh_{sd}^\dag\qH\qq=\qzero, \|\qq\|^2=1.\notag
 \eea
The optimal solution and the optimal objective value of (\ref{eqn:qq}) are given by, respectively,
  \be
    \qq_{ZF} =  \frac{\Pi_{\qH^\dag\qh_{sd}}^\bot\qH_{ed}^\dag\qh_{se}}{\|\Pi_{\qH^\dag\qh_{sd}}^\bot\qH_{ed}^\dag\qh_{se}\|} \mbox{~and~}
    \|\Pi_{\qH^\dag\qh_{sd}}^\bot\qH_{ed}^\dag\qh_{se}\|^2,
  \ee
where $\Pi_\qX^\bot \triangleq \qI - \qX(\qX^\dag\qX)^{-1}\qX^\dag$ denotes orthogonal projection onto the orthogonal complement of the column space
of $\qX$ and has the property $\Pi_\qX^\bot\Pi_\qX^\bot=\Pi_\qX^\bot$.  The resulting secrecy rate is then written as
 \bea
    R_{S}
    &=&\max\left(0, \log\left(1+  p_s \|\qh_{sd}\|^2 \right) \right.\notag\\
    &&\left. - \log\left(1+\frac{p_s\|\qh_{se}\|^2}{1+ \frac{p_d\|\Pi_{\qH^\dag\qh_{sd}}^\bot\qH_{ed}^\dag\qh_{se}\|^2}{\|\qh_{se}\|^2}}\right)\right).
\eea
 The optimal power allocation with a total power constraint $p_s + p_d=P_T$ can then be optimized similar to (\ref{eqn:fpt}).

\section{Transmission design with CDI}
In the previous sections, we assumed that the eavesdropper CSI is perfectly known. This information  can be available when $\tE$ is also an active user in
the network (unauthorized user) but  in general it is difficult to   obtain.   In this section, we study the case in which both $\tS$ and $\tD$ have
perfect CSI,  $\qh_{sd}$, but only CDI on $\tE$. For simplicity, we assume that the elements of $\qh_{se}$ and $\qH_{ed}$ are zero-mean independent
and identically distributed (i.i.d.) Gaussian random variables with variances $\sigma_s^2$ and $\sigma^2_d$, respectively.

\subsection{Expected Secrecy Rate}
  With CDI only, we first aim to maximize the ergodic secrecy rate \cite{Ulukus-07}, i.e.,
   \bea\label{eqn:rate:expectation}
    \max_{\qQ,p_s, \|\qr\|=1} && {\tt E}_{\qh_{se}, \qH_{ed}}\left( \log_2 \left(  1+ \frac{p_s|\qr^\dag\qh_{sd}|^2}
    {1+  \rho\qr^\dag\qH\qQ\qH^\dag\qr }\right)\right.\notag\\
    &&\left. -\log_2\left( {1+\frac{p_s\|\qh_{se}\|^2}{1+ \frac{\qh_{se}^\dag\qH_{ed}\qQ\qH_{ed}^\dag\qh_{se}}{\|\qh_{se}\|^2}}} \right)\right) \\
    \mbox{s.t.} && \qQ\succeq \qzero,  p_s \ge 0,  ~ p_s + \tr(\qQ)\le P_T.\notag
 \eea
 Due to the lack of instantaneous knowledge about the channels to $\tE$, we consider a suboptimal MRC receiver based only on local information:
 $\qr=\frac{\qh_{sd}}{\|\qh_{sd}\|}$ and  $\qQ$ is chosen as $\qQ = \frac{p_d \qW\qW^\dag}{M_t-1}$, where $\qW$
is an orthogonal basis of the null space
 of $\qH^\dag\qr$, i.e., $\qr^\dag\qH\qW=\qzero, \qW \in \mathbb{C}^{M_t\times (M_t-1)}, \qW^\dag \qW = \qI_{(M_t-1)}$. Then (\ref{eqn:rate:expectation}) becomes
   \bea\label{eqn:rate:expectation:Q}
    \max_{p_d,p_s} && {\tt E}_{\qh_{se}, \qH_{ed}}\left( \log_2 \left(  1+  p_s\|\qh_{sd}\|^2 \right) \right.\notag\\
    &&\left.-\log_2\left( {1+\frac{p_s\|\qh_{se}\|^2}{1+ p_d \frac{\qh_{se}^\dag\qH_{ed}\qW\qW^\dag\qH_{ed}^\dag\qh_{se}}{(M_t-1)\|\qh_{se}\|^2}}} \right)\right) \\
    \mbox{s.t.} && p_s \ge 0,  ~ p_s +p_d\le P_T.\notag
 \eea
 The expectation is still difficult to evaluate, so we optimize its approximation below by taking expectation operations on each individual random terms:
    \bea\label{eqn:rate:expectation:UB}
    \max_{p_d,p_s} &&  \log_2 \left(  1+  p_s\|\qh_{sd}\|^2 \right)
    \\
    &&  -\log_2\left( {1+\frac{p_s {\tt E}_{\qh_{se}}(\|\qh_{se}\|^2)}{1+ p_d {\tt E}_{\qh_{se}, \qH_{ed}}
    \left[\frac{\qh_{se}^\dag\qH_{ed}\qW\qW^\dag\qH_{ed}^\dag\qh_{se}}{(M_t-1)\|\qh_{se}\|^2}\right]}} \right)\notag\\
    \mbox{s.t.} && p_s \ge 0,  ~ p_s +p_d\le P_T.\notag
 \eea
 Notice that this approximation provides neither an upper nor a lower bound of the original problem (\ref{eqn:rate:expectation:Q}). Its effect
 will be evaluated in Fig. \ref{fig:expected:rate:SNR}.

 We find the following lemma is useful to solve (\ref{eqn:rate:expectation:UB}).
 \begin{lemma}\label{lem:x}
    $X=\frac{\qh_{se}^\dag\qH_{ed}\qW\qW^\dag\qH_{ed}^\dag\qh_{se}}{\|\qh_{se}\|^2}$ is a central chi-square random variable with $2(M_t-1)$ degrees of
    freedom.  
 \end{lemma}
 \begin{proof}

    Suppose a matrix decomposition $\frac{\qh_{se}}{\|\qh_{se}\|}=\qU\qd$,
    where $\qU$ is a unitary matrix and  $\qd$ is a zero vector except its first element being 1.
    Define eigenvalue decomposition of $\qW=\qU_w\qD_w\qU_w^\dag$, where $\qU_w$ is a unitary matrix and
 $\qD_w$ is a diagonal matrix with all diagonal entries being 1 except one element being 0.

      Then $X=\frac{\qh_{se}^\dag\qH_{ed}\qW\qW^\dag\qH_{ed}^\dag\qh_{se}}{\|\qh_{se}\|^2}=\qd^\dag\qU^\dag\qH_{ed} \qU_w\qD_w\qU_w^\dag\qH_{ed}^\dag\qU\qd$.
   Statistically $X$ is identical to
   \be\qd^\dag\qH_{ed}\qD_w\qH_{ed}\qd = \sum_{n=1}^{M_t-1} \|\qH_{ed}(1,n)\|^2,\ee
   where $\qH_{ed}(1,n)$ denotes  $\qH_{ed}$'s   $(1,n)$-th element. This completes the proof.

 \end{proof}

 Using Lemma \ref{lem:x}, $\frac{{\tt }E[X]}{(M_t-1)}=1$, so we have the following formulation:
    \bea\label{eqn:rate:expectation:UB2}
    \max_{p_d,p_s} &&  \log_2 \left(  1+  p_s\|\qh_{sd}\|^2 \right)
    -\log_2\left( {1+\frac{p_s M_e \sigma^2_s}{1+ p_d \sigma^2_d }} \right)\\
    \mbox{s.t.} && p_s \ge 0,  ~ p_s +p_d\le P_T.\notag
 \eea
 Its solution can be found using similar procedures to solve (\ref{eqn:fpt}).

 \subsection{Outage Secrecy Rate}
  Now we take a different design criterion  regarding the available CDI {  for a slow fading channel},
  and we aim to maximize the $\epsilon$-outage secrecy rate $r$ defined by
  \bea
       && \mbox{Prob}_{\qh_{se}, \qH_{ed}}\left(\log_2 \left(  1+  p_s|\qr^\dag\qh_{sd}|^2
     \right)  \right. \notag \\ && \left. -\log_2\left( {1+\frac{p_s\|\qh_{se}\|^2}
    {1+ \frac{p_d x}{M_t-1}}}\right)\le r \right) = \epsilon,
  \eea
   where we have used the assumption $\qQ=\frac{p_d \qW\qW^\dag}{M_t-1}$ and $x$ is a random variable defined in Lemma \ref{lem:x}.
   The outage secrecy rate maximization problem can be formulated as
    \bea\label{eqn:rate:outage:UB3}
    \max_{p_s, p_d} && r\\ 
    \mbox{s.t.} &&
    \mbox{Prob}_{\qh_{se}, \qH_{ed}}
    \left( \frac{1+  p_s\|\qh_{sd}\|^2 }{1+\frac{p_s  \|\qh_{se}\|^2 }{1+ \frac{p_d x}{M_t-1}}}\le 2^r\right)\le \epsilon,\notag\\
    &&p_s \ge 0,  ~ p_s +p_d\le P_T.\notag
 \eea
We can use bisection method to find the optimal $r$. The remaining problem is how to calculate the outage probability, which is
 rewritten below:
 \bea
   && \mbox{Prob}_{\qh_{se}, \qH_{ed}}
    \left( \frac{1+  p_s\|\qh_{sd}\|^2 }{1+\frac{p_s  \|\qh_{se}\|^2 }{1+ \frac{p_d x}{M_t-1} }}\le 2^r\right)\\
    &=&\mbox{Prob}_{\qh_{se}, \qH_{ed}}
    \left(  \frac{1+  p_s\|\qh_{sd}\|^2}{2^r} \le 1+\frac{p_s  \|\qh_{se}\|^2 }{1+ \frac{p_d x}{M_t-1}  }\right)\notag\\
   &=& \mbox{Prob}_{\qh_{se}, \qH_{ed}}
    \left(  \frac{p_s  \|\qh_{se}\|^2 }{1+\frac{p_d x}{M_t-1} }\ge  \frac{1+  p_s\|\qh_{sd}\|^2}{2^r}-1 \triangleq \alpha \right)\notag\\
  &=&   \mbox{Prob}_{\qh_{se}, \qH_{ed}}
    \left(   \frac{p_s}{\alpha}  \|\qh_{se}\|^2 - \frac{p_d x}{M_t-1}\ge  1 \right)\notag\\
 &=&   \mbox{Prob}_{\qh_{se}, \qH_{ed}}
    \left(  \gamma\triangleq \qz^\dag\qD\qz  \ge  1
    \right),  {~\mbox{and}~}\qz \sim\mathcal{CN}(0,\qI),\notag
 \eea
 where $\qD \triangleq \mbox{diag}\left(\left[\underbrace{\frac{p_s}{\alpha},\cdots, \frac{p_s}{\alpha}}_{M_e},
    \underbrace{  -p_d, \cdots, -p_d}_{M_t-1}\right]\right).$
It can be seen that $\gamma$ is an indefinite quadratic form in complex normal variables and the outage probability can be derived as
\cite{Raphaeli_96}
\bea &&\mbox{Prob}(\gamma\ge 1)=\frac{e^{-\frac{1}{\bar a}}}{{\bar a}^m{\bar b}^n(m-1)!(n-1)!}
\\&&\times\sum_{i=0}^{m-1}\sum_{k=0}^{m-i-1}\frac{(m-1)!(n+i-1)!}{i!(m-i-1)!k!}\left(\frac{\bar a+\bar b}{\bar a\bar b}\right)^{-(i+n)}{{\bar a}^{m-i-k}}, \notag\eea
where $\bar a \triangleq \frac{p_s}{\alpha}, \bar b \triangleq p_d, m \triangleq M_e, n\triangleq M_t-1.$

\section{Design of $\qQ$ with the Optimal linear MMSE receiver at both $\tD$ and $\tE$}


In previous sections, we have assumed that $\tE$ is not aware of the FD operation of $\tD$ and simply uses  an MRC   receiver $\qh_{se}^\dag$. However,
once $\tE$ learns that there is additional interference, it can adopt more advanced linear MMSE receiver:
 \be
    \qr_e = \frac{(\qH_{ed}\qQ\qH_{ed}^\dag+ \qI)^{-1}\qh_{se}}{\|(\qH_{ed}\qQ\qH_{ed}^\dag+ \qI)^{-1}\qh_{se} \|},
 \ee
  which, assuming that $\tD$ also uses the optimal linear MMSE receiver, leads to the secrecy rate
  \bea
    R_S^{'} &=& \max\left(0, \log\big(1+ P_s \qh_{sd}^\dag
    \big(\rho\qH\qQ\qH^\dag+ \qI \big)^{-1}\qh_{sd}  \big)\right. \notag\\
         &&\left. -\log\big(1+ P_s \qh_{se}^\dag\big(\rho\qH_{ed}\qQ\qH_{ed}^\dag+ \qI\big)^{-1}\qh_{se}  \big)\right).
  \eea
This results in the following secrecy rate maximization problem: \bea\label{eqn:MMSE} \max_{\qQ}&& F(\qQ) \triangleq \log\big(1+ P_s \qh_{sd}^\dag
    \big(\rho\qH\qQ\qH^\dag+ \qI \big)^{-1}\qh_{sd}  \big) \\
    && \qquad\qquad -  \log\big(1+ P_s \qh_{se}^\dag\big(\rho\qH_{ed}\qQ\qH_{ed}^\dag+ \qI\big)^{-1}\qh_{se}  \big) \notag\\
    \mathrm{s.t.} && \qQ\succeq \qzero, ~ \tr(\qQ)\le P_d.\notag
\eea For simplicity, we assume that perfect CSI is available. The problem of (\ref{eqn:MMSE}) is in general  not   convex and the optimal design of $\qQ$ to
maximize $R_S^{'}$ is a cumbersome optimization problem, so we propose to use the DC (difference of convex functions) programming \cite{An} to find a
stationary point. First, we express $F(\qQ)$ as a difference of two concave functions $f(\qQ)$ and $g(\qQ)$: \bea
 F(\qQ)&=&\log\det(\qI + P_s\qh_{sd}\qh_{sd}^\dag + \rho\qH\qQ\qH^\dag) \notag \\
&&  - \log\det(\qI + P_s\qh_{se}\qh_{se}^\dag + \rho\qH_{ed}\qQ\qH_{ed}^\dag)\notag\\
&&- \log\det (\qI + \rho\qH\qQ\qH^\dag) + \log\det(\qI + \rho\qH_{ed}\qQ\qH_{ed}^\dag) \notag \\
&\triangleq&  f(\qQ) - g(\qQ) \eea where \bea f(\qQ) &\triangleq & \log\det(\qI + P_s\qh_{sd}\qh_{sd}^\dag + \rho\qH\qQ\qH^\dag)\notag\\
&& +
\log\det(\qI + \rho\qH_{ed}\qQ\qH_{ed}^\dag), \\
g(\qQ) &\triangleq & \log\det (\qI + \rho\qH\qQ\qH^\dag)\notag\\
&& + \log\det(\qI + P_s\qh_{se}\qh_{se}^\dag + \rho\qH_{ed}\qQ\qH_{ed}^\dag). \eea
 The linearization of $g$ around the point $\qQ_k$ is
  \bea &&g_L(\qQ;\qQ_k) =\log\det (\qI + \rho\qH\qQ_k\qH^\dag) \\&&+
\log\det(\qI + P_s\qh_{se}\qh_{se}^\dag + \rho\qH_{ed}\qQ_k\qH_{ed}^\dag) \notag\\
&&\quad + \mathrm{Tr}\big(\rho \qH^\dag(\qI + \rho\qH\qQ_k\qH^\dag)^{-1}\qH(\qQ - \qQ_k) \big)
\notag\\
&&\quad + \mathrm{Tr}\big(\rho \qH_{ed}^\dag (\qI + P_s\qh_{se}\qh_{se}^\dag + \rho\qH_{ed}\qQ\qH_{ed}^\dag)^{-1}\qH_{ed} (\qQ - \qQ_k)\big).\notag
\eea DC programming is used to sequentially solve the following convex problem, $k=0, 1, \cdots$ \bea
\qQ_{k+1} &= &\arg \max_{\qQ}\ f(\qQ) - g_L(\qQ; \qQ_k) \label{DCP}\\
&&\mathrm{s.t.}\quad \qQ\succeq 0, \ \mathrm{Tr}(\qQ) \le 1. \notag \eea

To conclude this section, problem (\ref{eqn:MMSE}) can be solved by i) choosing an initial point $\qQ_0$; and ii) for $k=0, 1, \cdots$, solving
(\ref{DCP}) until the termination condition is met.

\section{Numerical Results}
Computer simulations are conducted to evaluate the performance of the proposed FD scheme.
 All channel entries are i.i.d. drawn from the Gaussian distribution $\mathcal{CN}(0,1)$.
Unless otherwise specified, it is assumed that $\tE$ is equipped with the same number of receive antennas as $\tD$, i.e., $M_e = M_r$, $M_t = M_r =2$
and $\rho=0.5$. The total transmit SNR, $P_T$ in dB, is used as power metric.

 We will compare the secrecy rate performance of the proposed FD scheme with the  baseline  HD system, in which the achievable secrecy capacity is expressed as:
 \be
    C_{S,HD} = \max\left(0, \log_2(1+ P_s\|\qh_{sd,HD}\|^2) - \log_2(1+P_s\|\qh_{se}\|^2)\right),
 \ee
where $\qh_{sd,HD}$ denotes the channel between $\tS$ and all antennas at $\tD$ in the HD mode.   We  maintain the same total power for the HD and
the FD systems.

 For the optimization with fixed receiver, we consider the following (non-optimized) MMSE receiver
 \be
    \qr = \frac{\left(\rho\qH \qH^\dag+ \qI\right)^{-1}\qh_{sd}  }{\|\left(\rho\qH \qH^\dag+ \qI\right)^{-1}\qh_{sd} \|}
 \ee
 which takes into account the self-interference channel as well as noise power.

{
}

In Fig. \ref{fig:rate:SNR:siso}, we  evaluate the achievable secrecy rate against total transmit SNR for the single-antenna case. We simulate both
cases with fixed equal power allocation between $\tS$ and $\tD$, and optimal power allocation. It is seen that both FD schemes outperform the HD
operation for transmit SNR greater than $ 10$ dB, and substantial secrecy rate gain is achieved in the high SNR region. The performance of  the HD scheme
saturates when the transmit SNR is higher than $25$ dB  while  ceiling effects for both FD schemes start to appear when the transmit SNR is $50$ dB.

In Fig. \ref{fig:rate:SNR}, we show the same results as those in Fig. \ref{fig:rate:SNR:siso} for the default multi-antenna setting.
Firstly, it is observed that for the HD mode, the secrecy rate saturates from very low SNR, because $\tE$ has the same number of antennas as $\tD$,
while $\tS$ has a single antenna and there is no external helper. On the other hand, for all FD schemes, the secrecy rate can keep increasing as the
transmit SNR increases without hitting a ceiling; this is due to the fact that multiple transmit antennas at $\tD$  help suppress self-interference
and generate jamming signal to $\tE$. It can be also seen that the joint $\tS-\tD$ power allocation only gives marginal performance gain in this
case. The optimal linear receiver at $\tD$ can improve the secrecy rate by approximately $10\%$ compared to a  fixed MMSE receiver. {  When an MMSE
receiver is used at $\tE$, the achievable secrecy rate starts to outperform the FD case only when the SNR is greater than $10$ dB, and the performance gain
is reduced compared to the case where $\tE$ uses simple MRC receiver.}

In Fig. \ref{fig:rate:rho}, we examine the impact of residual self-interference, $\rho$, when the total transmit SNR is $15$ dB. The HD mode does not
suffer from self-interference, therefore the secrecy rate remains constant. As expected, the secrecy rates of all FD schemes decrease as $\rho$
increases. All FD schemes outperform the HD scheme, even when $\rho$ is as high as $0.9$. It can  also be  seen that the performance gap between the
optimal linear receiver and the MMSE receiver becomes larger for higher $\rho$. This  is because the optimal linear receiver is much more effective
in compensating the higher self-interference.  {  When an MMSE receiver is used at $\tE$, much lower secrecy rate is achieved as compared to the case
in which  $\tE$ uses simple MRC receiver, but the FD approach  still outperforms the HD one.}

In Fig. \ref{fig:expected:rate:SNR}, we provide the results of ergodic secrecy rate against the total transmit SNR. The approximation of the ergodic
secrecy rate, which is used for power allocation, is also shown for comparison. Similar to  the perfect CSI case, the ergodic secrecy rate of the HD
scheme saturates from very low SNR, while using the FD scheme, the ergodic secrecy rate can increase without ceiling effect.

 In Fig. \ref{fig:outage:rate:SNR}, we plot the results of outage secrecy rate, where the target outage probability is $10\%$. It is seen that for the
 HD case, the achievable secrecy rate is very close to zero for all SNR regions.  Using the FD operation at $\tD$, the secrecy rate can be increased with
 SNR, even with fixed  power allocation at $\tS$ and $\tD$.

Finally in Fig. \ref{fig:rate:SNR:AS}, we investigate the impact of different combinations of transmit and receive antennas at $\tD$ assuming there
are   $M_r+M_t=4$ antennas in total. $\tE$ has $M_e=4$ receive antennas. We can see that $(M_r=2, M_t=2)$ provides the best performance, and this is
because the joint transmit and receive beamforming design can handle the self-interference efficiently. Another interesting observation is that
the case of $(M_r=3, M_t=1)$ greatly outperforms the case of $(M_r=1, M_t=3)$. This is due the fact that at $\tD$, the receiver design takes care of both useful signal and self-interference, while the transmitter design mainly aims to suppress self-interference and effectively jam $\tE$. In this sense, receiver design is
more important than the transmitter design, therefore more receive antennas can provide additional performance gain.


\section{Conclusions}
 In this paper, we have proposed a new self-protection scheme against passive eavesdropping achieved by the  FD operation at the destination. This is
 of particular interest when  the secrecy performance of the system is interference-limited and trusted external helpers are not available.
 To deal with LI and maximize the achievable secrecy rate, we have studied the optimal jamming covariance matrix at the destination and possible power allocation between the source
 and the destination with both perfect CSI and CDI. We have shown that the optimal jamming covariance matrix is rank-1 and can be found via an efficient 1-D search. In addition,  a low-complexity ZF
solution and the associated achievable secrecy rate have been derived in closed-form. Using the proposed FD scheme, the system is shown to be no
longer interference-limited, in contrast to  the HD case. Substantial performance gains are observed compared
 with the conventional HD operation at the destination.

 An interesting future direction is to study more sophisticated scenarios, where the eavesdropper knows the FD strategy employed at the destination,
and performs a similar FD operation; this is likely to be studied within the framework of non-cooperative game theory.

\section*{\sc Appendices}

\section*{A. Proof of Proposition \ref{prop:pd}}

\begin{proof}
For simplicity,  we define $a\triangleq P_s |h_{sd}|^2, b\triangleq \rho|h|^2, c \triangleq P_s|h_{se}|^2,$ and  $ d\triangleq |h_{ed}|^2$. Then
$f_\rho(x)$ can be re-expressed as
 \be
    f_\rho(x) = \frac{1+ \frac{a}{1+bx}}{1+\frac{c}{1+dx}}.
 \ee

 Setting its first-order  derivative of $f_\rho(x)$ to be zero leads to
 \be\label{eqn:fxzero}
 (cb-ad) x^2 + 2(c-a) x -\frac{a}{d}(1+c)+\frac{c}{b}(1+a)=0.
 \ee

  We first consider a trivial case that  $cb-ad=0$.  In this case, positive secrecy rate is possible only when $a>c$.   Consequently, the first-order derivative is negative, which indicates that
  $p_d^*(\rho) =0$ is the optimal solution.   In the following, we assume that $cb-ad\ne 0$.

 We   show that  $x_1(\rho)$ and $x_2(\rho)$ cannot be both positive, which requires
  \bea
   x_1(\rho)+ x_2(\rho)= \frac{ a-c}{cb-ad}>0, \\~~x_1(\rho) x_2(\rho) =  \frac{cd(1+a)-ab(1+c)}{bd(cb-ad)}>0.
  \eea
We assume the first equality is true, i.e., $x_1(\rho)+ x_2(\rho)= \frac{ a-c}{cb-ad}>0$, in the following, we show that the second one cannot be
satisfied.
 \begin{itemize}
    \item[1)] $a\ge c$:

    In this case, $cb-ad>0$ and $b>d$. We have
    $$cd(1+a)-ab(1+c)\le cb(1+a)-ab(1+c)=(c-a)b\le0.$$

   Thus, $x_1(\rho) x_2(\rho) =  \frac{cd(1+a)-ab(1+c)}{bd(cb-ad)}\le0$.

    \item[2)] $a<c$:

    In this case, $cb-ad<0$ and $b<d$. We have
    $$cd(1+a)-ab(1+c)> cb(1+a)-ab(1+c)=(c-a)b>0.$$
      Thus, $x_1(\rho) x_2(\rho) =  \frac{cd(1+a)-ab(1+c)}{bd(cb-ad)}<0$.
 \end{itemize}
  As a result, we know that $x_2(\rho)$ is the only possible positive root of $f_\rho'(x)=0$ if there exists one.

  When $\frac{ad}{cb} = \frac{\delta}{\rho}>1$ or $\rho\le \delta$,  $f_\rho^{'}(p_d)>0$ for $0\le p_d\le x_2(\rho)$, so $p_d^*(\rho) =\min(P_d, x_2(\rho))$;
 when  $\rho> \delta$, it can be verified that
 $f_\rho^{'}(p_d)<0$ for $0\le p_d\le x_2(\rho)$ and $f_\rho^{'}(p_d)>0$ for $p_d> x_2(\rho)$. So if $ x_2(\rho)\ge P_d$, $p_d^*(\rho) =0$; otherwise
 $p_d^*(\rho) =0$ or $p_d^*(\rho) =P_d$. The proof is completed.
\end{proof}

\section*{\sc Appendix B}
 In this appendix, we derive  the solution to the problem below:
\bea \label{eqn:opt:R:w}
    \max_{\qQ\succeq \qzero} &&   \qa^\dag\qQ\qa \\
    \mbox{s.t.} && \tr(\qQ\qR)=t, \tr(\qQ)=1,\notag
\eea where $\qQ$ is an $M_t\times M_t$ matrix variable, $\qR=\qH^\dag\qh_{sd}\qh_{sd}^\dag\qH - t\rho P_d \qH^\dag\qH $ is a Hermitian matrix which
is not positive semidefinite and $\qa=\qH_{ed}^\dag \qh_{se}$ is an $M_t\times 1$ vector. $t$ is a positive scalar.

 The dual problem is
 \bea
    \min_{\lambda_1,\lambda_2} && \lambda_1 t + \lambda_2  \\
    \mbox{s.t.} && \lambda_2\qI + \lambda_1\qR -\qa\qa^\dag\succeq \qzero.\notag
 \eea
 It can be checked that $\lambda_2>0$, otherwise, the secondary matrix inequality cannot be satisfied.
 It is also known that there is at least one rank-1 optimal solution $\qq$ that satisfies $\qQ=\qq\qq^\dag$.

 Due to the complementary slackness, we know  that
 \be
    (\lambda_2\qI + \lambda_1\qR -\qa\qa^\dag)\qq=\qzero.
 \ee

We first consider a special case that $\lambda_2\qI + \lambda_1\qR$ is positive semi-definite and its smallest eigenvalue is 0. 
Then we have the following inequality:
    \bea
    0&=& \lambda_{\min} (\lambda_2\qI + \lambda_1\qR -\qa \qa^\dag)\\
     &=& \min_{\|\qx\|=1} \qx^\dag(\lambda_2\qI + \lambda_1\qR-\qa \qa^\dag)\qx\\
       & =&\min_{\|\qx\|=1} (\qx^\dag (\lambda_2\qI + \lambda_1\qR)\qx -\qx^\dag\qa \qa^\dag \qx)\\
        &\le&\min_{\|\qx\|=1} (\qx^\dag (\lambda_2\qI + \lambda_1\qR)\qx + \lambda_{\max} (-\qa \qa^\dag))\\
        &=&\min_{\|\qx\|=1}  \qx^\dag (\lambda_2\qI + \lambda_1\qR)\qx \\
        &=& 0.
    \eea

Suppose one eigenvector that achieves $\min_{\|\qx\|=1}  \qx^\dag (\lambda_2\qI + \lambda_1\qR)\qx$ is $\qq_1$, i.e., $\qq_1^\dag(\lambda_2\qI +
\lambda_1\qR)\qq_1=0,$ then the equality sign is attained iff $\qq_1^\dag\qa=0$, which is a trivial case.

 In the following, we assume  $\lambda_2\qI + \lambda_1\qR\succ\qzero$. It is easy to verify that at the optimum, $\qa^\dag(\lambda_1\qR + \lambda_2\qI)^{-1}\qa=1$.  Suppose the eigenvalue
  decomposition of $\qR$ is $\qR=\qU\qD\qU^\dag, \qU \qU^\dag =\qI, \qD = \mbox{diag}(d_1, \cdots, d_{M_t})$, then
 $\qq$ can be expressed as
 \be\label{eqn:q}
    \qq = \qU (\bar\lambda_1\qD + \bar\lambda_2\qI)^{-1}\qb, \qb \triangleq \qU^\dag \qa=[b_1,\cdots, b_{M_t}]^T,
 \ee
 where $\bar\lambda_1= s\lambda_1, \bar\lambda_2= s \lambda_2, s = \frac{1}{\qb^\dag(\lambda_1\qD + \lambda_2\qI)^{-1}\qb}$. The remaining
 task is to identify the optimal $\bar\lambda_1$ and $\bar\lambda_2$.

 Substitute (\ref{eqn:q}) into the two equality constraints in (\ref{eqn:opt:R:w}), then we get the following two equations:
 \be
   \left\{\begin{array}{c}
        \sum_{i=1}^{M_r} |b_i|^2\frac{d_i}{(\bar\lambda_1 d_i + \bar\lambda_2)^2} = t\\
          \sum_{i=1}^{M_t} |b_i|^2\frac{1}{(\bar\lambda_1 d_i + \bar\lambda_2)^2} = 1,
        \end{array}\right.
\ee from which we can solve all solutions $(\bar\lambda_1,\bar\lambda_2)$  and choose the one which satisfies $\bar\lambda_2>0$ and returns the
maximum objective value $\bar\lambda_1 t + \bar\lambda_2 $.

 \begin{figure}[t]
  \centering
  \includegraphics[width=3in]{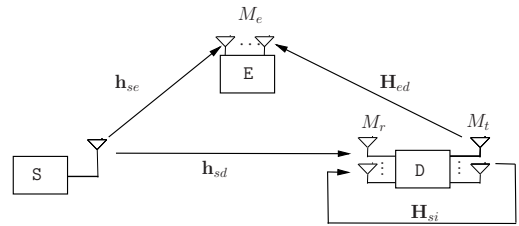}
  \vspace{-0.3cm}
  \caption{System model with FD receiver.}\label{fig:sys}
\end{figure}

\begin{figure}
  \centering
  \includegraphics[width=3in]{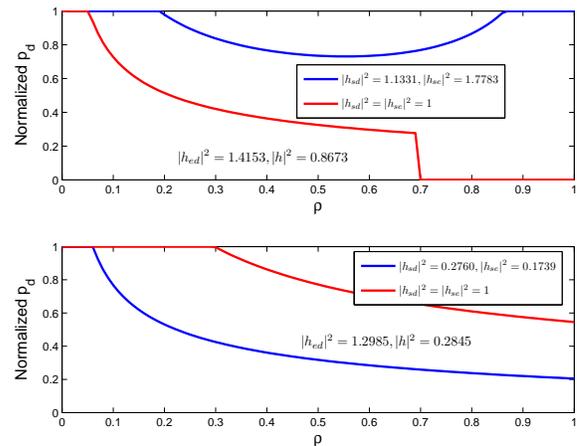}\\
  \vspace{-0.4cm}
  \caption{The impacts of $\rho$ on the optimal power allocation at $\tD$.}\label{fig:pd:nonmono}
\end{figure}

\begin{figure}[h]
  \centering
  \includegraphics[width=3in]{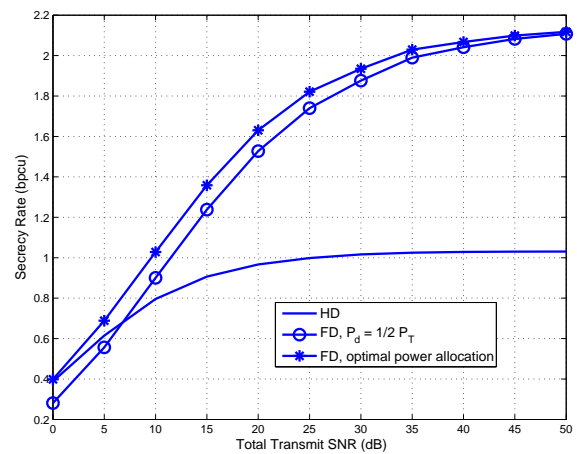}
  \caption{Secrecy rate vs. total transmit power for the single-antenna case.}\label{fig:rate:SNR:siso}
\end{figure}

 \begin{figure}[h]
  \centering
  \includegraphics[width=3in]{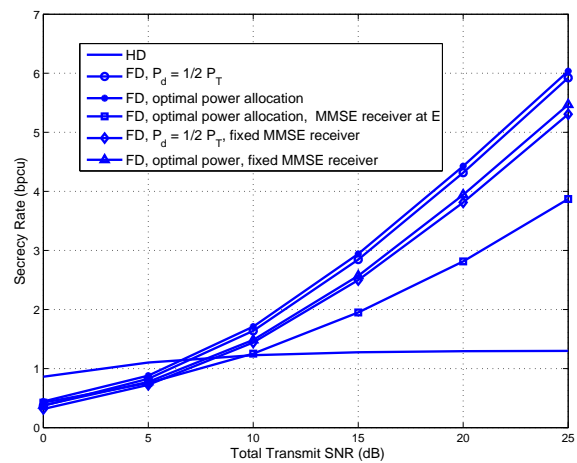}
  \caption{Secrecy rate vs. total transmit power for the multi-antenna case.}\label{fig:rate:SNR}
\end{figure}

 \begin{figure}[h]
  \centering
  \includegraphics[width=3.5in]{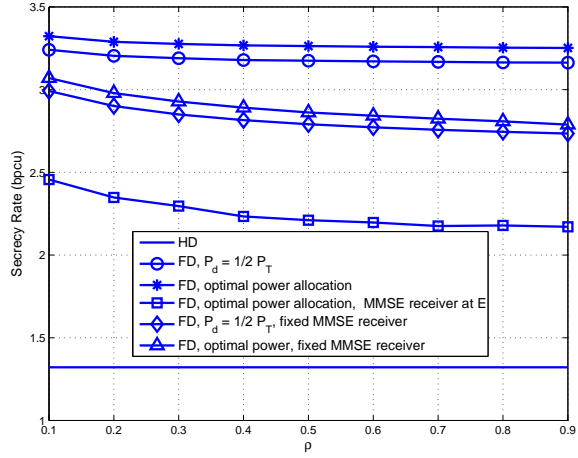}
  \caption{The effects of the LI channel strength $\rho$ on the achievable secrecy rate.}\label{fig:rate:rho}
\end{figure}

 \begin{figure}[h]
  \centering
  \includegraphics[width=3in]{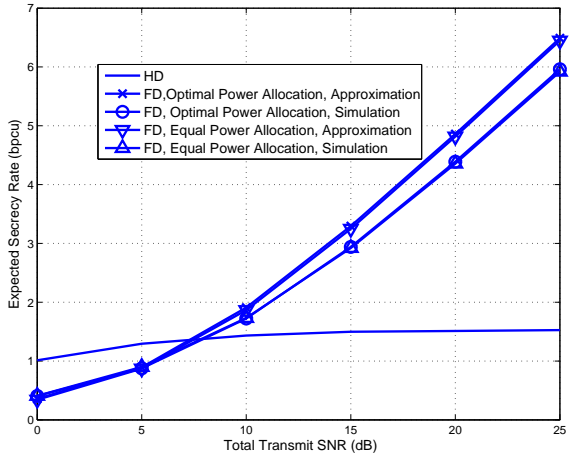}
  \caption{ Ergodic secrecy rate performance for the fast fading channel when only the CDI about $\tE$ is available.}\label{fig:expected:rate:SNR}
\end{figure}

 \begin{figure}[h]
  \centering
  \includegraphics[width=3in]{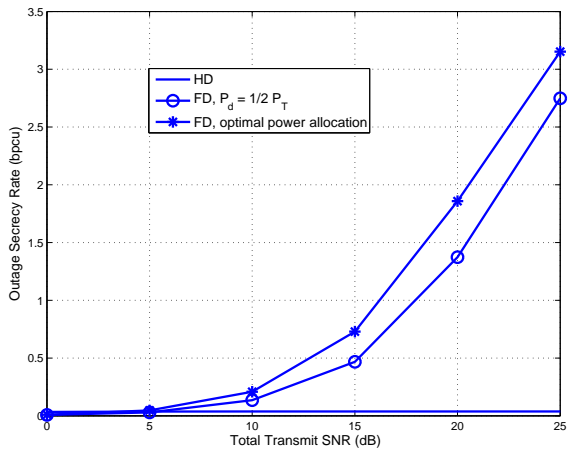}
  \caption{Outage secrecy rate  for the slow fading channel when only the CDI about $\tE$ is available. }\label{fig:outage:rate:SNR}
\end{figure}

 \begin{figure}[h]
  \centering
  \includegraphics[width=3in]{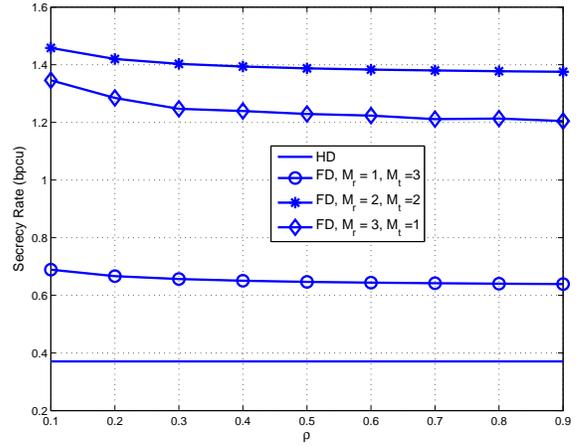}
  \caption{The impacts of $\rho$  and antenna configurations at $\tD$ on the achievable secrecy rate.}\label{fig:rate:SNR:AS}
\end{figure}

 \end{document}